\documentclass{amsart}
\usepackage{mathrsfs}

\newtheorem{theorem}{Theorem}

\newtheorem{proposition}{Proposition}[section]
\newtheorem{lemma}{Lemma}[section]
\newtheorem{corollary}{Corollary}[section]

\theoremstyle{definition}

\theoremstyle{remark}

\numberwithin{equation}{section}

\def\rd{\text{\rm d}}
\def\re{\text{\rm e}}
\def\ri{\text{\rm i}}

\def\Real{\text{\rm Re}}
\def\Imag{\text{\rm Im}}

\begin{document}

\title[Pad\'{e} approximants of random Stieltjes series]{Pad\'{e} approximants
of random Stieltjes series}

\author{Jens Marklof}
\address{School of Mathematics\\
        University of Bristol\\
        Bristol BS8 1TW, United Kingdom}
\email{j.marklof@bristol.ac.uk}

\author{Yves Tourigny}
\address{School of Mathematics\\
        University of Bristol\\
        Bristol BS8 1TW, United Kingdom}
\email{y.tourigny@bristol.ac.uk}

\author{Lech Wo{\l}owski}
\address{School of Mathematics\\
        University of Bristol\\
        Bristol BS8 1TW, United Kingdom}
\email{lech@ucdavis-alumni.com}

\thanks{The authors gratefully acknowledge the support
of the Engineering and Physical Sciences Research Council
(United Kingdom) under Grant GR/S87461/01 and of the Leverhulme Trust under a Philip
Leverhulme Prize (JM). The second author would also like to thank Professor Alain Comtet
for his hospitality during a visit to the Laboratoire de Physique Th\'{e}orique et 
Mod\`{e}les Statistiques, Universit\'{e} de Paris-Sud (Orsay), and for
many useful discussions while some of this work was carried out.}





\begin{abstract}
We consider the random continued fraction
$$
S(t) := \cfrac{1}{s_1 + \cfrac{t}{s_2 + \cfrac{t}{s_3 + \cdots}}}\,, \quad t \in {\mathbb C} \backslash {\mathbb R}_-\,,
$$
where the $s_n$ are independent random variables 
with the same gamma distribution. Every realisation
of the sequence 
defines
a Stieltjes function that can be expressed as
$$
S(t) = \int_{0}^\infty \frac{\sigma (\rd x)}{1+xt}\,,\quad t \in {\mathbb C} \backslash {\mathbb R}_-\,,
$$
for some measure $\sigma$ on the positive half-line. 
We study the convergence of the finite truncations of the continued fraction or, equivalently,
of the diagonal Pad\'{e} approximants of the function $S$. By using the Dyson--Schmidt method
for an equivalent one-dimensional disordered system, and
the results of Marklof {\em et al.} (2005), we obtain explicit formulae (in terms of modified Bessel functions)
for the almost-sure rate of convergence of these approximants,
and for the almost-sure distribution of their poles.
\end{abstract}

\maketitle

\section{Introduction}
Let $\mathbf s = (s_1,s_2,\ldots)$ be a sequence of positive real numbers and 
consider the analytic continued
fraction
\begin{equation}
S(t) := \cfrac{1}{s_1 + \cfrac{t}{s_2 + \cfrac{t}{s_3 + \cdots}}}\,, \quad t \in {\mathbb C} \backslash {\mathbb R}_-\,.
\label{continuedFraction}
\end{equation}
This continued fraction defines a {\em Stieltjes function}; it can be represented in integral
form as
\begin{equation}
S(t) = \int_0^\infty \frac{\sigma (\rd x)}{1+t x}
\label{stieltjesIntegral}
\end{equation}
for some measure $\sigma$ supported on the non-negative half-line such that the {\em moments}
\begin{equation}
m_n := \int_0^\infty x^n \sigma(\rd x), \quad n \in {\mathbb N},
\label{moments}
\end{equation}
exist. 
By an obvious use of the geometric series, every Stieltjes
function can be expanded formally in powers of $t$:
\begin{equation}
S(t) \sim \sum_{j=0}^\infty m_j \,(-t)^j \quad \text{as $t \rightarrow 0+$}\,.
\label{series}
\end{equation}
Hence $S$
is the {\em moment generating function} of the measure $\sigma$. 
It is a well-known fact of great practical importance that, given the first $n$ of the moments, one
may construct the rational function
\begin{equation}
S_n(t) := \frac{P_n(t)}{Q_n(t)} =
\cfrac{1}{s_1 + \cfrac{t}{s_2+ \cdots + \cfrac{t}{s_{n}}}}\,, \quad t \in {\mathbb C} \backslash {\mathbb R}_-\,,
\label{convergent}
\end{equation}
where 
$$
\deg P_n = \begin{cases}
n/2-1 & \text{if $n$ is even} \\
(n-1)/2 & \text{if $n$ is odd}
\end{cases}
\quad \text{and} \quad
\deg Q_n = \begin{cases}
n/2 & \text{if $n$ is even} \\
(n-1)/2 & \text{if $n$ is odd}
\end{cases}\,.
$$
This truncation of the continued fraction (\ref{continuedFraction})
has a MacLaurin expansion whose $n$th partial sum agrees
with that of the series (\ref{series}). Hence $S_n$ is a diagonal (if $n$ is odd), 
or near-diagonal (if $n$ is even) {\em Pad\'{e}
approximant} of $S$.

Now suppose that the $s_n$ are independent positive random variables with the same distribution,
say $\mu$.
We shall consider the following questions: 
\begin{enumerate}
\item What are the almost-sure analytic properties
of these Stieltjes functions? 
\item What is the almost-sure leading asymptotic behaviour of the error $S(t)-S_n(t)$
as $n \rightarrow \infty$? 
\end{enumerate}

These questions are of interest because Pad\'{e} approximation is widely used in applied mathematics
as a practical means of accelerating the convergence of the partial sums of series obtained
by perturbation methods. As pointed out by Bender \& Orszag (1978), the consideration of many particular
cases where the $s_n$ are {\em deterministic} reveals a wide range of large-$n$ behaviours.
Our motivation for studying the random case is to gain some insight into the asymptotic behaviour
of Pad\'{e} approximation in the ``generic'' case.

Our study of diagonal Pad\'{e} approximation reduces to aspects of
the large-$n$ behaviour of the denominators $Q_n$.
For this reason, as in the deterministic case,
the cornerstone of the analysis
is the three-term recurrence relation
\begin{equation}
Q_{n+1} = t Q_{n-1} + s_{n+1} Q_n\,, \quad Q_{-1} = 0,\;\; Q_0 = 1\,.
\label{threeTermRecurrence}
\end{equation}
(The $P_n$ satisfy the same recurrence relation, albeit with different
initial conditions.)
This recurrence relation makes a link between Pad\'{e} approximation and a
rich set of other mathematical entities, such as orthogonal polynomials,
products of random matrices, and discrete Schr\"odinger-like operators.
By exploiting results that are well-known in these related fields, one
may obtain--- for a very large class of distributions $\mu$ of the coefficients
$s_n$--- some partial answers to the questions stated earlier.
Our contribution in the present paper is to elaborate the  
particular case where $\mu$ is the gamma distribution. More precisely, we obtain
explicit formulae (in terms of Bessel functions) for the leading term in the 
asymptotic behaviour of the error of Pad\'{e} approximation and for the asymptotic
distribution of the poles, as well as the location of the essential spectrum of
the measure $\sigma$.

In the remainder of this introductory section, we describe briefly 
the key ideas underlying the
analysis. Then we summarise our main results in the form of a theorem. 

\subsection{The moment problem}
The Stieltjes moment problem is, given a sequence $\{m_n\}_{n \in {\mathbb N}}$,
to determine whether or not there exists a measure $\sigma$ such that equation (\ref{moments})
holds for every $n$. 
Historically, mathematical objects such as the analytic continued 
fraction (\ref{continuedFraction}), orthogonal polynomials
and Pad\'{e} approximants were introduced as tools in the study of
this moment problem (Akhiezer 1961; Nikishin \& Sorokin 1988; Simon 1998). 
Stieltjes (1894) showed that a necessary
and sufficient condition for the {\em existence}
of a measure $\sigma$ with the prescribed moments is 
\begin{equation}
\forall \; n \in {\mathbb Z}_+\,, \quad s_n > 0\,.
\label{stieltjesExistenceCondition}
\end{equation}
He also showed that 
\begin{equation}
\sum_{n=1}^\infty s_n = \infty
\label{stieltjesUniquenessCondition}
\end{equation}
is a necessary and sufficient condition for the {\em uniqueness} of the measure.

Let us assume that condition (\ref{stieltjesExistenceCondition}) holds and describe in very
broad terms one way 
of ``reconstructing'' $\sigma$ from
its moments;
see Akhiezer (1961) and Nikishin \& Sorokin (1988) for a detailed treatment. 

Recall that $x'$ is a {\em point of increase} of the measure $\sigma$ if
$$
\forall\; \varepsilon > 0, \quad \int_{\max \{ x'-\varepsilon,0\}}^{x'+\varepsilon}
\sigma (\rd x) >0\,.
$$
The {\em spectrum} of $\sigma$ is the set of its points of increase and will be denoted
$\text{spec}(\sigma)$.
For $x > 0$, we shall
denote by $\delta_{x}$ the probability measure on ${\mathbb R}_+$
whose only point of increase is $x$.

Given the $m_n$, we may define an
inner product, say $(\cdot,\cdot)_{\mathbf m}$, on the space of polynomials as follows: if
$p$ and $q$ are two polynomials with coefficients $p_i$ and $q_i$ respectively, then
$$
(p,q)_{\mathbf m} := \sum_{i,j} m_{i+j} \,p_i \,q_j\,.
$$
Knowing the moments, we may 
compute the
$s_n$ and the $S_n$. We remark that $P_{2n}$ and $Q_{2n}$ are 
polynomials of degree $n-1$ and $n$ respectively. 
Set
\begin{equation}
\psi_n (\lambda) := 
\sqrt{s_{2n+1}} \,\lambda^n Q_{2n} (-1/\lambda)\,.
\label{psiDefinition}
\end{equation}
Then the fact that $S_{2n}$ matches the moment generating series
to $O(t^{2n})$ implies that $\psi_n$ is orthogonal
to every polynomial of degree less than $n$, in the sense of
the inner product $(\cdot,\cdot)_{\mathbf m}$. It follows (see
Akhiezer 1961, Chapter 1) that the roots of $\psi_n$ are
simple and lie in ${\mathbb R}_+$; denote them by
$$
0 \le  \lambda_{n,1} < \lambda_{n,2} < \cdots < \lambda_{n,n} < \infty\,.
$$
Gaussian quadrature then defines
a discrete measure
\begin{equation}
\sigma_n := \sum_{j=1}^n \sigma_{n,j} \delta_{\lambda_{n,j}}\,, \quad \text{where}\;\;
\sigma_{n,j} = \left ( \sum_{\ell=0}^{n-1} \psi_\ell^2(\lambda_{n,j}) \right )^{-1}\,,
\label{quadratureFormula}
\end{equation}
that converges weakly to a measure $\sigma$
that solves the moment problem. In particular, the
inner product $(\cdot,\cdot)_{\mathbf m}$ coincides with the inner
product in $L_\sigma^2 ({\mathbb R}_+)$.

The moment problem can also be approached from the point of view of operator theory.
The recurrence relation (\ref{threeTermRecurrence}) for the $Q_n$ implies the
following recurrence relation for the $\psi_n$:
\begin{align*}
v_0 \psi_0 + h_0 \psi_1 &= \lambda \psi_0 \quad \text{if $n=0$}\,, \\
h_{n-1} \psi_{n-1} + v_n \psi_n + h_{n} \psi_{n+1} &= \lambda \psi_n \quad \text{if $n \in {\mathbb Z}_+$}
\end{align*}
where
\begin{equation}
\notag
v_n = \begin{cases}
\frac{1}{s_1 s_2} & \text{if $n=0$} \\
\frac{1}{s_{2n+1}} \left ( \frac{1}{s_{2n}} + \frac{1}{s_{2n+2}} \right ) & \text{if $n \in {\mathbb Z}_+$}
\end{cases}
\,,
\quad
h_n = 
\frac{1}{s_{2n+2} \sqrt{s_{2n+1} s_{2n+3}}}\,, \;\; n \in {\mathbb N}
\,.
\label{differenceCoefficients}
\end{equation}
These numbers may be used to define a 
certain Jacobi operator, say ${\mathscr J}$,
with a domain contained in the Hilbert space $\ell^2({\mathbb N})$, as follows:
first, we consider sequences $\boldsymbol{\xi}
= (\xi_0,\,\xi_1,\,\xi_2,\,\ldots)$ with only finitely many terms and set
\begin{equation}
\left ( {\mathscr J} \boldsymbol{\xi} \right )_n := \begin{cases}
v_0 \xi_0 + h_0 \xi_1 & \text{if $n=0$} \\
h_{n-1} \xi_{n-1} + v_n \xi_n + h_n \xi_{n+1} & \text{otherwise}
\end{cases}\,.
\label{definitionOfJ}
\end{equation}
Given condition (\ref{stieltjesUniquenessCondition}), it is then possible to extend this 
definition uniquely to obtain an essentially self-adjoint operator; we use the same symbol ${\mathscr J}$ to
refer to this extension. 
It may then be proved
that the moments of the spectral measure of the operator ${\mathscr J}$ 
are precisely the $m_n$, and so this spectral measure coincides with $\sigma$; see Nikishin \& Sorokin (1988).

\subsection{The density of states}
\label{densityOfStatesSubsection}

Consider the finite-dimensional truncation
$$
{\mathscr J}_n := \begin{pmatrix}
v_0 & h_0 & 0 & 0 & \cdots & 0 & 0 & 0 \\
h_0 & v_1 & h_1 & 0 & \cdots & 0 & 0 & 0 \\
0 & h_1 & v_2 & h_2 & \cdots & 0 & 0 & 0 \\
\vdots & \ddots & \ddots & \ddots & \ddots & \ddots & \ddots & \vdots  \\
0 & 0 & 0 & 0 & \cdots & h_{n-3} & v_{n-2} & h_{n-2} \\
0 & 0 & 0 & 0 & \cdots & 0 & h_{n-2} & v_{n-1} 
\end{pmatrix}
$$
of the operator $\mathscr{J}$. The spectrum of $\mathscr{J}_n$
is the set of zeroes of the polynomial $\psi_n$
defined by equation (\ref{psiDefinition}).
By the Chebyshev--Markov--Stieltjes theorem (see Nikishin \& Sorokin 1988, \S 2.8), between any two zeroes
of $\psi_n$, there is a point
of increase of $\sigma$; so we are led to studying the distribution of the $\lambda_{n,j}$.

Define a measure $\kappa_n$ on ${\mathbb R}_+$ by
\begin{equation}
\kappa_n := \frac{1}{n} \sum_{j=1}^n \delta_{\lambda_{n,j}}\,.
\label{discreteDensityOfStates}
\end{equation}
This measure is the normalised eigenvalue counting measure of the matrix $\mathscr{J}_n$.
Indeed, we have
$$
N_n (\lambda) := \frac{\# \left \{j\,:\; \lambda_{n,j} < \lambda \right \}}{n}
= \int_0^\lambda \kappa_n (\rd \lambda')\,.
$$

The normalised counting measure $\kappa_n$ has a weak limit, say $\kappa$, as $n \rightarrow \infty$, and so
there is a function $N$, called the {\em integrated density of states} of $\mathscr{J}$,
defined by
$$
N(\lambda) := \int_0^\lambda \kappa (\rd \lambda) = \lim_{n \rightarrow \infty} N_n (\lambda)  \,.
$$
If $\kappa$
is absolutely continuous, one can also speak of the {\em density of states}, say $\varrho$, defined
by
\begin{equation}
\kappa(\rd \lambda) = \varrho(\lambda) \,\rd \lambda\,.
\label{densityOfStates}
\end{equation}
Although
the measures $\kappa$ and $\sigma$ may be
very different, their essential spectra are the same. In the context of Pad\'{e} approximation, the
integrated density of states describes the distribution of the poles of the approximants.

\subsection{Krein's string}
\label{KreinSubsection}
There is an interpretation, due to Krein, of the spectrum of the operator $\mathscr{J}$ in terms of the
characteristic frequencies of a vibrating string (Akhiezer 1961). Consider a weightless, infinite,
perfectly elastic string, tied at one endpoint $x=0$, along which some beads are distributed. Let
$s_{2n}$ be the mass of the $n$th bead, and denote by $(x_n,y_n)$ its position in the $xy$-plane.
We assume that the $x_n$ are fixed and given by the recurrence relation
$$
x_{n+1} = x_n + s_{2n+1}\,, \quad x_0 = 0\,.
$$
For a string of uniform unit tension, the small vertical motion is then described by the
discrete wave equation
\begin{equation}
s_{2n} \ddot{y}_n = \frac{y_{n+1}-y_n}{s_{2n+1}} - \frac{y_n-y_{n-1}}{s_{2n-1}}\,, 
\quad n \in {\mathbb Z}_+\,.
\label{waveEquation}
\end{equation}
To study the characteristic frequencies of the string, we set
$$
y_n = \eta_n \sin \left ( \omega t \right ) \quad \text{and} \quad
\xi_n = \frac{\eta_{n+1}-\eta_n}{s_{2n+1}}\,.
$$
Then equation (\ref{waveEquation}) reduces to
$$
\frac{1}{s_1 s_2} \xi_1 - \frac{1}{s_1 s_2} \xi_0 = - \omega^2 \xi_0 
$$
and 
$$
\frac{1}{s_{2n+1} s_{2n+2}} \xi_{n+1} - \frac{1}{s_{2n+1}} 
\left ( \frac{1}{s_{2n}}+\frac{1}{s_{2n+2}} \right ) \xi_n
+ \frac{1}{s_{2n} s_{2n+1}} \xi_{n-1} = - \omega^2 \xi_n
$$
where $n \ge 1$. Comparing this with the definitions of $\mathscr{J}$ and $\psi_n$
given earlier, it is readily seen that
$$
\xi_n = 
\frac{(-1)^n \xi_0}{\sqrt{s_{2n+1}}} \psi_n ( \omega^2 ) 
\,. 
$$

\subsection{The complex Lyapunov exponent}
\label{lyapunovSubsection}
Dyson (1953) developed a method for studying the characteristic frequencies of the one-dimensional
disordered chain
\begin{equation}
\ddot{y}_n = c_{2n-1} (y_{n+1}-y_n )- c_{2n-2} (y_n-y_{n-1})\,,
\quad n \in {\mathbb Z}_+\,.
\label{dysonChain}
\end{equation}
Here, $c_{2n-1}$ and $c_{2n-2}$
are the ratios of the elastic modulus of the $n$th spring and
of the mass of the two particles attached to it. Disorder may be modelled in many ways; for instance by
assuming that the $c_n$ are independent and identically distributed.
The approach was
later simplified by Schmidt (1957) and applied
to the tight-binding Anderson model for a one-dimensional crystal with impurities. 

Luck (1992) gives a very readable, well-motivated account
of the Dyson--Schmidt approach; in brief, it builds on the intimate connection between
second-order difference equations, continued fractions and Markov chains.
For our purpose, it will be convenient
to work with 
the random difference equation
\begin{equation}
u_{n+1} - u_{n-1} = \frac{s_{n+1}}{\sqrt{t}}\,u_n\,, \quad n=0\,,1\,,2\,,\ldots
\label{differenceEquation}
\end{equation}
where 
$t$ is a parameter in ${\mathbb C} \backslash {\mathbb R}_-$, and $\sqrt{\cdot}$ is the branch of the square
root function defined on ${\mathbb C} \backslash {\mathbb R}_-$ that returns a number with a non-negative
real part.
We note the obvious
\begin{lemma} 
For every $t \in {\mathbb C} \backslash {\mathbb R}_-$,
$$
Q_n (t) = \left ( \sqrt{t} \right )^n u_n\,, 
$$
where $u_n$ solves the difference equation (\ref{differenceEquation}) with $u_{-1}=0$ and $u_0=1$.
\label{denominatorLemma}
\end{lemma}
The relevant
continued fraction is
\begin{equation}
Z  := \sqrt{t} \,S(t) = \cfrac{1}{\frac{s_1}{\sqrt{t}} + \cfrac{1}{\frac{s_2}{\sqrt{t}} + \cfrac{1}{\frac{s_3}{\sqrt{t}} + \cdots}}}
\label{complexContinuedFraction}
\end{equation}
and we write
\begin{equation}
Z_n = \sqrt{t}\,S_n
\label{Zn}
\end{equation}
for its truncation.
Let $u_{-1}$ and $u_0$ be complex random variables. Then equation (\ref{differenceEquation})
defines a sequence of general term $u_n$ by recurrence. The distribution $\nu_{\mu}$ of 
the random variable $Z$ 
is a stationary distribution for the Markov chain
\begin{equation}
\hat{\mathbf Z}  = \left ( \hat{Z}_0,\,\hat{Z}_1,\,\hat{Z}_2,\, \ldots \right ) \,, \quad \text{where}
\;\; \hat{Z}_n := \frac{u_{n-1}}{u_n}\,.
\label{markovChain}
\end{equation}
(In the terminology of iterated random maps, the random variables
$Z_n$ and $\hat{Z}_n$ are, respectively, the {\em backward} and {\em forward iterates}
associated with the continued fraction; when $u_{-1}=0$ and $u_0=1$, they have the same distribution, but their asymptotic
behaviours are very different; see Diaconis \& Freedman 1999.)

It follows that the growth of the $u_n$ may be quantified by means of the
{\em complex Lyapunov exponent} defined by
\begin{equation}
\Lambda_{\mu} (t) := -\int_{\mathbb C} \ln z \,\nu_{\mu} (\rd z)\,,
\label{characteristicExponent}
\end{equation}
where $\ln \cdot$ denotes the principal branch of the logarithm. 
Indeed, standard results from the theory of Markov chains imply
that, if $\hat{\mathbf Z}$ has a unique stationary distribution, then 
\begin{equation}
\frac{\ln u_n}{n} = \frac{\ln u_0}{n} - \frac{1}{n} \sum_{j=1}^n \ln \hat{Z}_j \xrightarrow[n \rightarrow \infty]{} \Lambda_{\mu} (t)
\label{markovChainTheorem}
\end{equation}
for almost every realisation of ${\mathbf s}$,
independently of the choice of $\hat{Z}_0$ (Breiman 1960; Furstenberg 1963; Meyn \& Tweedie 1993). 
In particular, we have the formula 
$$    
\lim_{n \rightarrow \infty} \frac{\ln |u_n|}{n} = \Real \,\left [ \Lambda_{\mu} (t)\right ]\,.
$$
Equation (\ref{markovChainTheorem})
is also central to
the study of the integrated density of states of the operator $\mathscr{J}$ introduced earlier (Dyson 1953;
Schmidt 1957; Luck 1992).
We shall see that, under a very mild assumption on the distribution $\mu$ of the $s_n$,
$$
N(\lambda) = - \frac{2}{\pi} \, \Imag \left [
\Lambda_{\mu} \left ( -1/\lambda + \ri 0 + \right ) \right ]\,.
$$

\subsection{Furstenberg's theorem}
\label{furstenbergSubsection}
In order to carry out this programme, we shall also make use of the connection between
the Markov chain $\hat{\mathbf Z}$ 
and the product of random matrices
\begin{equation}
{\mathcal U}_n := {\mathcal A}_n {\mathcal A}_{n-1} \cdots {\mathcal A}_1\,, \quad n=0\,,1\,,2\,,\ldots\,,
\label{productOfMatrices}
\end{equation}
where
\begin{equation}
{\mathcal A}_n := \begin{pmatrix}
0 & 1 \\
1 & \frac{s_n}{\sqrt{t}}
\end{pmatrix}\,.
\label{definitionOfA}
\end{equation}
The distribution $\mu$ from which the $s_n$ are drawn induces, via equation (\ref{definitionOfA}), a distribution $\tilde{\mu}$
on the group of unimodular $2 \times 2$ matrices. The fundamental
results of Furstenberg \& Kesten (1960)
and Furstenberg (1963) --- which are commonly referred to as `Furstenberg's theorem'---
imply in particular that,
under very mild assumptions on $\mu$, there is a unique measure $\tilde{\nu}_{\tilde{\mu}}$ on the
group of unimodular matrices that is invariant under the
action of the random matrix (\ref{definitionOfA}). Furthermore, the number
\begin{equation}
\gamma_{\tilde{\mu}} := \frac{1}{n} {\mathbb E} \left ( \ln \left | {\mathcal U}_n 
\right | \right ),
\label{realLyapunovExponent}
\end{equation}
which quantifies the growth of the product ${\mathcal U}_n$ 
and is independent of the choice of matrix norm $| \cdot |$, may be shown to be {\em strictly positive}.
These results are of great relevance to our problem for two reasons: firstly,
any invariant measure $\tilde{\nu}_{\tilde{\mu}}$
yields a measure $\nu_{\mu}$ that is stationary for the Markov chain  $\hat{\mathbf Z}$, and
vice-versa; secondly,
$$
\gamma_{\tilde{\mu}} = \Real \,\left [ \Lambda_{\mu} (t)\right ]\,.
$$
Hence we deduce at once the uniqueness of the measure $\nu_{\mu}$, as well as the exponential
growth of the $u_n$: 

\begin{proposition}
Let the $s_n$ be independent 
random variables in ${\mathbb R}_+$ with a common distribution $\mu$
that has at least two points of increase. 
Let $\{u_n\}_{n \in {\mathbb N}}$ be the sequence
defined by the recurrence (\ref{differenceEquation}).
Suppose that
$$
\int_{{\mathbb R}_+} s^\varepsilon \,\mu(\rd s) < \infty
$$
for some $\varepsilon > 0$.
Then, for almost every realisation of the
sequence ${\mathbf s}$, the following holds independently of the starting value $u_0 \ne 0$:
For Lebesgue-almost every $t \in {\mathbb C} \backslash {\mathbb R}_-$,
$$
\lim_{n \rightarrow \infty} \frac{\ln u_n}{n} = \Lambda_{\mu} (t) \quad \text{and} \quad
\Real \left [ \Lambda_{\mu} (t) \right ] > 0\,.
$$
Furthermore, if
we set $t=-x+\ri 0\pm$, then for Lebesgue-almost every $x>0$, 
$$
\lim_{n \rightarrow \infty} \frac{\ln u_n}{n} = \Lambda_{\mu} (-x+\ri 0\pm) \quad \text{and} \quad
\Real \left [ \Lambda_{\mu} (-x+\ri 0 \pm ) \right ] > 0\,.
$$
\label{furstenbergProposition}
\end{proposition}

The proof of this proposition is provided in Appendix
\ref{furstenbergAppendix}.

\subsection{The gamma distribution: statement of the main result}
\label{resultsSubsection}
By using such machinery, we are able, for a very wide class of distributions $\mu$,
to deduce the almost-sure exponential nature of the convergence of diagonal Pad\'{e}
approximation and also to deduce
the almost-sure singularity of the measure $\sigma$. A more quantitative study
requires the calculation of the complex Lyapunov exponent,
but there are very few known instances where it can be expressed
in terms of familiar functions. 

In his seminal paper on the disordered chain (\ref{dysonChain}), Dyson studied in some detail the
particular case where the $c_n$ are independent and gamma-distributed. Dyson found
the invariant distribution of the continued fraction 
\begin{equation}
\notag
\cfrac{c_0 t}{1+ \cfrac{c_1 t}{1+\cfrac{c_2 t}{1+\cdots}}}
\end{equation}
in the particular case where $t>0$; 
he then used analytic continuation to obtain an expansion for the 
complex Lyapunov exponent at $t<0$, and hence for the distribution of the 
characteristic frequencies. 
Dyson's continued 
fraction is not equivalent to ours (compare equations (\ref{dysonChain}) and (\ref{waveEquation})), and so his analytical results do not transfer
to our problem. However, 
the continued fraction (\ref{complexContinuedFraction}) 
with independent, gamma-distributed $s_n$, i.e. where
\begin{equation}
\mu (\rd x) := \frac{1}{b^a \,\Gamma(a)} \,x^{a-1} \exp(-x/b) \,\rd x\,, \quad a,\,b >0\,.
\label{gammaDistribution}
\end{equation}
was examined also by Letac and Seshadri (1983); they obtained
the probability distribution $\nu_\mu$ 
and found an explicit formula for the corresponding
Lyapunov exponent for $t>0$. In a recent paper, we generalised this result by finding
$\nu_{\mu}$ and the real part of the complex Lyapunov exponent
for {\em every} complex $t$ (see Marklof {\em et al.} 2005); a straightforward extension of these calculations leads to
the remarkably simple formula
\begin{equation}
 \Lambda_{\mu}(t)
= 
\partial_{a} \ln \left [ K_{a}\left(\frac{2\sqrt{t}}{b} \right) \right ]\,.
\label{characteristicExponentFormula}
\end{equation}
In this expression, 
$\partial_a$ denotes differentiation with respect to $a$, and $K_a$ is the modified Bessel function
of the second kind (see Appendix \ref{characteristicExponentAppendix}). The following summarises the key results
of this paper.
\begin{theorem}
Suppose that the $s_n$ are independent draws from
the gamma distribution with parameters $a>0$ and $b>0$. Then,
for
almost every realisation of the sequence ${\mathbf s}$, the following holds:
\begin{enumerate}
\item The density of states is given explicitly by the formula
$$
\varrho (\lambda) = -\frac{2}{\pi^2 \lambda} \,\partial_a \left [ 
\frac{1}{J_a^2 \left ( \frac{2}{b \sqrt{\lambda}} \right ) + Y_a^2 \left ( \frac{2}{b \sqrt{\lambda}} \right )} \right ]\,.
$$
\item 
$\text{\em spec}(\sigma) = [0,\,\infty)$ and its absolutely continuous part is empty.

\item For Lebesgue-almost every $t \in {\mathbb C} \backslash {\mathbb R}_-$,
$$
\lim_{n \rightarrow \infty} \frac{\ln \left | S(t)-S_n(t) \right |}{n} = 
-2 \,\partial_a \ln \left | K_a \left ( \frac{2 \sqrt{t}}{b} \right ) \right |
\,.
$$
\end{enumerate}
\label{mainTheorem}
\end{theorem}

\subsection{Relation to other work}

As stated earlier, our focus in the present paper
is on the performance of diagonal Pad\'{e} approximation, viewed as a method
of summing some random series. Related questions have been considered in the past, in different contexts. 
Foster \& Pitcher (1974) study the convergence of random 
$T$-fractions; these are continued fraction expansions which are in a one-to-one
correspondence with the space of formal power series, but whose convergents are
not Pad\'{e} approximants. Foster \& Pitcher (1974) show that, under very
general conditions on the distribution of the coefficients, the difference between
two successive convergents tends to zero exponentially fast, and that the exponent
is twice the Lyapunov exponent associated with an infinite product of random matrices.
Geronimo (1993) studies the random measure (on the unit circle) generated by a three-term
recurrence relation with random, identically distributed coefficients;
he shows the positivity of the corresponding Lyapunov exponent and deduces
that the random measure is singular with respect to the Lebesgue measure. This list is not
exhaustive; see also Csordas {\em et al.} (1973) and Mannion (1993).

The question of the nature of the measure $\sigma$ has a counterpart in the
theory of disordered systems which has been studied extensively in the context
of Anderson localisation. For example,
the tight-binding Anderson model uses a discretised version of the Schr\"{o}dinger equation 
with a potential that takes random, identically-distributed values at every point
in a doubly-infinite lattice. The resulting operator 
has a
second-order finite-difference form like that
of the operator ${\mathscr J}$---in which, more precisely, $h_n = 1$ and the $v_n$
are independent and identically distributed--- but it acts on sequences
in $\ell^2(\mathbb Z)$.
For a very wide choice of the distribution of the potential values $v_n$, the Lyapunov
exponent of the discretised Schr\"{o}dinger operator
is strictly positive, so that, by Ishii's formula,
the absolutely continuous spectrum is empty. A more refined study (see
for instance Carmona \& Lacroix 1990 and Pastur \& Figotin 1992) reveals that
these operators have a {\em pure point} spectrum
--- that is, the spectrum is the closure of the discrete spectrum. Such
operators are said to exhibit the {\em localisation property} because, for this type of spectrum,
the generalised eigenfunctions decay exponentially fast as $|n| \rightarrow \infty$.
The rigorous extension of such detailed results to the semi-infinite case
would involve technicalities which are outside the scope of the present paper.

Our analysis exploits a number of ideas and techniques found in these earlier studies.
{\em We view our main contribution as that of exhibiting an interesting example of a class of random
Stieltjes functions for which the leading behaviour of the error of diagonal Pad\'{e} approximation,
and the density of states of the corresponding Jacobi operator,
are given explicitly in terms of special functions.}

The remainder of the paper is devoted to a detailed proof of theorem \ref{mainTheorem}: the first statement follows immediately from Dyson's formula for the density of states, which is
derived in  
\S \ref{densityOfStatesSection}. In \S \ref{singularitySection}, we deduce  
the singularity of the measure $\sigma$ from the positivity of the real part of the Lyapunov exponent.
In \S \ref{convergenceSection}, we show that the error of diagonal Pad\'{e} approximation is
inversely proportional to the square of the $u_n$; this yields
the third statement in the theorem.
Finally, in \S \ref{numericalSection}, we provide a
numerical illustration of our results.

\section{The formula for the density of states}
\label{densityOfStatesSection}

So that we can use proposition \ref{furstenbergProposition},
we shall henceforth suppose that
the $s_n$ are independent draws from a distribution $\mu$ on ${\mathbb R}_+$ such that
\begin{enumerate}
\item $\mu$ has at least two points of increase.
\item There exists $\varepsilon > 0$ such that
$$
\int_{{\mathbb R}_+} s^\varepsilon \mu (\rd s) < \infty\,.
$$
\end{enumerate}
To avoid needless repetitions, we shall not mention these particular
assumptions explicitly
again in the statement of the intermediate results leading to
theorem \ref{dysonTheorem}. 
We begin by relating the growth of the $\psi_n$ to the complex Lyapunov exponent.

\begin{lemma}
For almost every realisation of the sequence ${\mathbf s}$, we have,  for Lebesgue-almost every $\lambda \in {\mathbb C} \backslash {\mathbb R}_+$,
$$
\lim_{n \rightarrow \infty} \frac{\ln \psi_n (\lambda)}{n} = \ri \pi + 2 \Lambda_{\mu} \left ( -1/\lambda \right )
$$
and, for Lebesgue-almost every $\lambda \in {\mathbb R}_+$,
$$
\lim_{n \rightarrow \infty} \frac{\ln \psi_n (\lambda + \ri 0 \pm )}{n} = 
\ri \pi + 2 \Lambda_{\mu} \left ( -1/\lambda + \ri 0 \pm \right ) \,.
$$
\label{asymptoticsOfPsiLemma}
\end{lemma}

\begin{proof}
Let $\lambda \in {\mathbb C} \backslash {\mathbb R}_+$. By definition,
$$
\psi_n (\lambda ) = \sqrt{s_{2n+1}} \,\lambda^n Q_{2n} (-1/\lambda) \,.
$$
Hence, by lemma \ref{denominatorLemma},
\begin{equation}
\notag
\psi_n (\lambda) = \sqrt{s_{2n+1}}\, \lambda^n \left ( \sqrt{-1/\lambda}\right )^{2n} u_{2n} \\
= (-1)^n \sqrt{s_{2n+1}}\, u_{2 n}\,, 
\end{equation}
where $u_n$ solves the difference equation (\ref{differenceEquation}) with $t=-1/\lambda$.
This yields
$$
\frac{\ln \psi_n (\lambda)}{n} = \frac{\ln s_{2n+1}}{2n}  + \ri \pi + 2 \frac{\ln u_{2n}}{2 n}
\,.
$$
The first statement in the proposition then follows from corollary \ref{markovChainCorollary}.  
The proof of the second statement is identical.
\end{proof}

Next, we examine the implications of the lemma for the
distribution of the $\lambda_{n,j}$. 
By virtue of the recurrence relation satisfied by the $\psi_n$, we can write
$$
\psi_n (\lambda) =  \sqrt{s_{2n+1}} \left ( \prod_{j=1}^{2n} s_j \right )\,( \lambda-\lambda_{n,1}) \cdots ( \lambda-\lambda_{n,n})\,.
$$
Let
$$
{E}_n := \left \{ \lambda_{n,j} : 1 \le j \le n \right \}
$$
and let $\lambda \notin {E}_n$.
Then
\begin{equation}
\frac{\ln \psi_n (\lambda)}{n} = \frac{\ln s_{2n+1}}{2n} + 
\frac{1}{n} \sum_{j=1}^{2n} \ln s_j + \frac{1}{n} \sum_{j=1}^n \ln (\lambda-\lambda_{n,j}) \,.
\label{logOfPsi}
\end{equation}

\begin{proposition}
Suppose that 
$$
\int_{{\mathbb R}_+} \left | \ln s \right |\,\mu (\rd s) < \infty\,.
$$
Then,
for almost every realisation of the sequence $\mathbf s$, for Lebesgue-almost every
$\lambda \in {\mathbb C} \backslash {\mathbb R}_+$,
$$ 
\int_0^\infty \ln | \lambda - \lambda' | \,\kappa_n ( \rd \lambda')
\xrightarrow[n \rightarrow \infty]{} 
2\,\Real \left [ \Lambda_{\mu} \left ( -1/\lambda \right ) \right ]
- 2 \int_{{\mathbb R}_+} \,\ln s \,\mu (\rd s)\,.
$$
\label{goldsheidKhoruzhenkoProposition}
\end{proposition}

\begin{proof}
Take the real part in equation (\ref{logOfPsi}). Then 
\begin{multline}
\notag
\int_0^\infty \ln | \lambda - \lambda' | \,\kappa_n ( \rd \lambda') = 
\frac{1}{n} \sum_{j=1}^n \ln | \lambda-\lambda_{n,j} | = 
\frac{\ln | \psi_n (\lambda)|}{n} \\
 - \frac{\ln s_{2n+1}}{2 n} - \frac{1}{n} \sum_{j=1}^{2n} \ln s_j \,.
\end{multline}
By lemma \ref{asymptoticsOfPsiLemma}, the first term on the right tends to
$$
2\,\Real \left [ \Lambda_{\mu} \left ( -1/\lambda \right ) \right ]\,,
$$
the second term tends to zero and,
by the ergodic theorem,  the third term 
tends to
$$
2 \int_{{\mathbb R}_+} \,\ln s \,\mu (\rd s)\,.
$$
\end{proof}

\begin{corollary}
Under the same assumption, for almost every realisation of the
sequence $\mathbf s$, the sequence $\{ \kappa_n \}_{n \in {\mathbb N}}$
has a weak limit, say $\kappa$, which is a probability measure on ${\mathbb R}_+$. In particular, 
$$
\lim_{n \rightarrow \infty} N_n(\lambda) = N (\lambda) := \int_0^\lambda \kappa (\rd \lambda')\,.
$$
\label{goldsheidKhoruzhenkoCorollary}
\end{corollary}

\begin{proof}
The proof is a specialisation of that given by Goldsheid \& Khoruzhenko (2005) in the
more general case of a non-Hermitian Jacobi matrix. 
See Appendix \ref{goldsheidAppendix} for the details. 
\end{proof}

\begin{theorem}
Let the $s_n$ be independent 
random variables in ${\mathbb R}_+$ with a common distribution $\mu$
that has at least two points of increase. Suppose also that
$$
\int_{{\mathbb R}_+} \left | \ln s \right | \,\mu ( \rd s) < \infty
\quad \text{and} \quad
\int_{{\mathbb R}_+} s^\varepsilon \,\mu ( \rd s) < \infty 
$$
for some $\varepsilon > 0$.
Then, for almost every realisation of the sequence ${\mathbf s}$, for Lebesgue-almost
every $\lambda \in {\mathbb R}_+$,
$$
N(\lambda) = - \frac{2}{\pi} \, \Imag \left [
\Lambda_{\mu} \left ( -1/\lambda + \ri 0 + \right ) \right ]\,.
$$
\label{dysonTheorem}
\end{theorem}

\begin{proof}
Let $\lambda \in {\mathbb R}_+ \backslash {E}_n$. Then
\begin{multline}
\notag
\ln \psi_n (\lambda + \ri 0 \pm ) - \frac{\ln s_{2n+1}}{2n} -
\frac{1}{n} \sum_{j=1}^{2n} \ln s_j = \sum_{j=1}^n \ln (\lambda-\lambda_{n,j}+ \ri 0 \pm ) \\
= \sum_{\lambda_{n,j}<\lambda} \ln (\lambda-\lambda_{n,j}+\ri 0 \pm)
+ \sum_{\lambda_{n,j}>\lambda} \ln (\lambda-\lambda_{n,j}+\ri 0 \pm) \\
=  \sum_{\lambda_{n,j}<\lambda} \ln | \lambda-\lambda_{n,j} |
+ \sum_{\lambda_{n,j}>\lambda} \left ( \pm \ri \pi + \ln | \lambda-\lambda_{n,j} | \right ) \\
=  \sum_{j=1}^n \ln | \lambda-\lambda_{n,j} | \pm \sum_{\lambda_{n,j} > \lambda} \ri \pi\,.
\end{multline}
Hence, we have the identity
\begin{multline}
\frac{\ln \psi_n (\lambda + \ri 0 \pm)}{n} = \frac{\ln s_{2n+1}}{2n} + 
\frac{1}{n} \sum_{j=1}^{2n} \ln s_j \\
+ \int_0^\infty \ln | \lambda - \lambda' | \,\kappa_n ( \rd \lambda')
\pm \ri \pi \left [ 1-N_n(\lambda) \right ]\,.
\label{dysonIdentity}
\end{multline}
Take the imaginary part.
The result then follows from lemma \ref{asymptoticsOfPsiLemma} and corollary 
\ref{goldsheidKhoruzhenkoCorollary}.
\end{proof}

\begin{corollary}
Suppose that the $s_n$ are independent and gamma-distributed with parameters $a$ and $b$. Then,
for $\lambda >0$, we have the following formula for the density of states:
$$
\varrho (\lambda) := N'(\lambda) =
-\frac{2}{\pi^2 \lambda} \,\partial_a \left [ 
\frac{1}{J_a^2 \left ( \frac{2}{b \sqrt{\lambda}} \right ) + Y_a^2 \left ( \frac{2}{b \sqrt{\lambda}} \right )} \right ]\,.
$$
\label{dysonCorollary}
\end{corollary}

\begin{proof}
Let $\lambda > 0$ and set
$$
z = \frac{2}{b \sqrt{\lambda}}\,.
$$ 
For $t=-1/\lambda+ \ri 0+$, we have
$$
K_a \left ( \frac{2 \sqrt{t}}{b} \right )
= K_a ( \ri z ) = -\frac{\pi}{2} \,\re^{-\ri a \frac{\pi}{2}} \,\left [Y_a (z)+
 \ri J_a (z) \right ]\,.
$$
Hence
$$
\partial_a K_a \left ( \frac{2 \sqrt{t}}{b} \right )
= - \ri \frac{\pi}{2} \, K_a \left ( \frac{2 \sqrt{t}}{b} \right ) - \frac{\pi}{2} \re^{-\ri a \frac{\pi}{2}}
\left [ \partial_a Y_a (z) + \ri \partial_a J_a (z) \right ]
$$
and so 
\begin{multline}
\notag
\Lambda_{\mu} (t) = \frac{\partial_a K_a \left ( \frac{2 \sqrt{t}}{b} \right )}{K_a \left ( \frac{2 \sqrt{t}}{b} \right )} =  \overline{K_a \left ( \frac{2 \sqrt{t}}{b} \right )} \frac{\partial_a K_a \left ( \frac{2 \sqrt{t}}{b} \right )}{\left | K_a \left ( \frac{2 \sqrt{t}}{b} \right ) \right |^2} \\
= - \ri \frac{\pi}{2} + \ri \frac{Y_a (z) \partial_a J_a(z) - J_a (z) \partial_a Y_a (z)}{J_a(z)^2+Y_a(z)^2} + \frac{J_a(z) \partial_a J_a(z)+ Y_a(z) \partial_a Y_a(z)}{J_a(z)^2+Y_a(z)^2}\,.
\end{multline}
We deduce the formulae
\begin{equation}
\Real \left [ \Lambda_{\mu} \left ( -\frac{1}{\lambda} + \ri 0+ \right ) \right ]
= \frac{J_a(z) \partial_a J_a(z)+ Y_a(z) \partial_a Y_a(z)}{J_a(z)^2+Y_a(z)^2} 
\label{realPartOflyapunovExponentOnTheNegativeAxis}
\end{equation}
and
\begin{equation}
\Imag \left [ \Lambda_{\mu} \left ( -\frac{1}{\lambda} + \ri 0+ \right ) \right ]
= - \frac{\pi}{2} + \frac{Y_a (z) \partial_a J_a(z) - J_a (z) \partial_a Y_a (z)}{J_a(z)^2+Y_a(z)^2}\,.
\label{imaginaryPartOflyapunovExponentOnTheNegativeAxis}
\end{equation}
The proposition, together with the last equation, then yield
\begin{equation}
\notag
1-N(\lambda) = \frac{2}{\pi} 
\frac{Y_a (z) \partial_a J_a(z) - J_a(z) \partial_a Y_a(z)}{J_a(z)^2 + Y_a(z)^2} \\
= - \frac{2}{\pi} \partial_a  \left [ \arctan \frac{Y_a (z)}{J_a(z)} \right ]\,.
\end{equation}
Differentiating both sides with respect to $\lambda$, we find
\begin{equation}
\notag
-\varrho (\lambda)= - \frac{2}{\pi} \frac{\rd z}{\rd \lambda} \frac{\rd}{\rd z}
\partial_a  \left [ \arctan \frac{Y_a (z)}{J_a(z)} \right ] =
\frac{2}{\pi} \frac{1}{b \lambda^{3/2}} \frac{\rd}{\rd z}
\partial_a  \left [ \arctan \frac{Y_a (z)}{J_a(z)} \right ]\,.
\end{equation}
We obtain the desired result by changing the order of differentiation on the right-hand side, and
making use of the identity
$$
J_a(z) Y_a ' (z)  - Y_a(z) J_a'(z) = \frac{2}{\pi z}\,.
$$
\end{proof}

\begin{corollary}
Under the same assumption, for almost every realisation of the sequence ${\mathbf s}$, 
$$
\text{\em spec} ( \sigma) = [0,\,\infty)\,.
$$
\label{densityOfStatesCorollary}
\end{corollary}
\begin{proof}
By differentiating the identity (see Watson 1966, \S 13.73)
$$
J_a^2 (z)+Y_a^2(z)
= \frac{8}{\pi^2} \int_0^\infty K_0 \left ( 2 z \sinh t \right ) \,\cosh (2 at) \,\rd t
$$
with respect to $a$, we deduce that $\varrho$ is strictly positive.
\end{proof}

\section{Singularity of the spectrum}
\label{singularitySection}

Corollary \ref{densityOfStatesCorollary}
implies in particular that the radius of convergence
of the generating series of the moments
is zero almost surely; in other words, the random Stieltjes functions that we have
constructed are not analytic at the origin.

The problem of determining the {\em nature} of the spectrum is more delicate.
Every measure may be decomposed
into three disjoint parts: its absolutely continuous, singular continuous and discrete parts,
denoted by $\sigma_{ac}$, $\sigma_{sc}$ and $\sigma_d$ respectively. 
Ishii (1973) and Yoshioka (1973) showed
that the spectrum of the absolutely continuous part 
is given by the formula
\begin{equation}
\text{spec} \left ( \sigma_{ac} \right ) = \overline{\left \{\lambda >0\,:\;
\lim_{t \rightarrow -\frac{1}{\lambda}+\ri 0+} \Real\,\Lambda_{\mu} (t) = 0 \right \}}\,.
\label{ishiiFormula}
\end{equation}
This result may be established by examining the resolvent of the Jacobi operator $\mathscr{J}$.
The following result is essentially equivalent.

\begin{proposition}
Let the $s_n$ be independent 
random variables in ${\mathbb R}_+$ with a common distribution $\mu$
that has at least two points of increase.  Assume that
$$
\int_{{\mathbb R}_+} \left | \ln s \right |\,\mu (\text{\em d} s) < \infty \quad \text{and} \quad
\int_{{\mathbb R}_+} s^\varepsilon\,\mu (\text{\em d} s) < \infty
$$
for some $\varepsilon > 0$. Then, for almost every realisation
of the sequence ${\mathbf s}$, 
$$
\text{\em spec} (\sigma_{ac}) = \emptyset\,.
$$
\label{essentialSpectrumProposition}
\end{proposition}
\begin{proof}
By lemma \ref{asymptoticsOfPsiLemma} and corollary \ref{positivityCorollary}, for almost every 
realisation of ${\mathbf s}$, for Lebesgue-almost every $\lambda \in {\mathbb R}_+$, we have
\begin{equation}
\lim_{n \rightarrow \infty} \frac{\ln \left | \psi_n (\lambda) \right |}{n} > 0\,.
\label{exponentialGrowth}
\end{equation}
Now, let $\eta >0$ and consider the set
$$
{\mathbb S}_\eta := \left \{ \lambda \in {\mathbb R}_+\,:\;\psi_n (\lambda) = o \left ( \sqrt{n} \,[ \ln n ]^{1+\eta} \right ) \;\; \text{as $n \rightarrow \infty$}
\right \}\,.
$$
By equation (\ref{exponentialGrowth}),
this set has Lebesgue measure zero almost surely. On the other hand, it follows from the Men'shov--Rademacher theorem (see Nikishin \& Sorokin 1988, proposition 8.3) that
for $\sigma$-almost $\lambda \in {\mathbb R}_+$, 
$$
\psi_n (\lambda) = o \left ( \sqrt{n} \,[\ln n]^{1+\eta} \right ) \;\; \text{as $n \rightarrow \infty$}\,.
$$
So, almost surely,
for every $\sigma$-measurable set $A$,
\begin{multline}
\notag
\int_{A} \sigma (\rd \lambda ) = \int_{A \cap {\mathbb S}_\eta} \sigma (\rd \lambda) =
\int_{A \cap {\mathbb S}_\eta} \sigma_{ac} (\rd \lambda) + \int_{A \cap {\mathbb S}_\eta} \sigma_{sc} (\rd \lambda)
+\int_{A \cap {\mathbb S}_\eta} \sigma_d (\rd \lambda)\\
=\int_{A \cap {\mathbb S}_\eta} \sigma_{sc} (\rd \lambda)
+\int_{A \cap {\mathbb S}_\eta} \sigma_d (\rd \lambda)
= \int_{A} \sigma_{sc} (\rd \lambda)
+\int_{A} \sigma_d (\rd \lambda)
\,.
\end{multline}
\end{proof}

\section{The rate of convergence}
\label{convergenceSection} 
Let $\{u_n\}_{n \in {\mathbb N}}$ be the sequence
defined by the recurrence relation (\ref{differenceEquation}) with the
starting values $u_{-1}=0$ and $u_0=1$. Also, denote by $\mathscr{T}$ the 
shift operator on the space of complex sequences, i.e.
$$
\mathscr{T} \left ( s_1,s_2,\ldots \right ) = \left (s_2,s_3,\ldots \right )\,.
$$ 
In order to emphasise the dependence of the continued fraction (\ref{complexContinuedFraction}) on the
sequence $\mathbf s$, we shall sometimes write $Z(\mathbf s)$ and $Z_n (\mathbf s)$ instead
of $Z$ and $Z_n$.
We have the following convenient representation of the
error:
\begin{lemma}
$$
Z (\mathbf s )-Z_n (\mathbf s)
= \frac{(-1)^n}{u_n} \left (\prod_{j=0}^n Z \left ( \mathscr{T}^j {\mathbf s} \right )\right )\,.
$$
\label{errorLemma}
\end{lemma}

\begin{proof}
By using the recurrence relations (\ref{threeTermRecurrence}) and the identity
$$
Z \left ({\mathscr T}^{n} \mathbf s \right ) = \frac{1}{\frac{s_{n+1}}{\sqrt{t}}+ Z \left ({\mathscr T}^{n+1} \mathbf s \right )}\,,
$$
it is straightforward
to show (by induction on $n$) that
$$
Z \left ( {\mathscr T}^n {\mathbf s} \right ) = - \frac{1}{\sqrt{t}}
\frac{Q_n S-P_n}{Q_{n-1} S-P_{n-1}}\,.
$$
Then
\begin{multline}
\notag
\prod_{j=0}^n Z \left ({\mathscr T}^j {\mathbf s} \right ) = (-1/\sqrt{t})^{n+1}
\prod_{j=0}^n \frac{Q_{j}S-P_j}{Q_{j-1}S-P_{j-1}} \\
= (-1/\sqrt{t})^{n+1} \frac{Q_n S-P_n}{Q_{-1} S - P_{-1}} 
= (-1/\sqrt{t})^n Q_n \sqrt{t} \left ( S - S_n \right )\,.
\end{multline}
Lemma \ref{denominatorLemma} then yields the desired result.
\end{proof}

\begin{theorem}
Let the $s_n$ be positive independent random variables with a common distribution $\mu$
that has at least two points of increase. Suppose also that there exists $\varepsilon >0$ such that
$$
\int_{{\mathbb R}_+} s^\varepsilon\,\mu (\text{\em d} s) < \infty\,.
$$
Then, for almost every realisation of
the sequence ${\mathbf s}$, for Lebesgue-almost every $t \in {\mathbb C} \backslash {\mathbb R}_-$,
$$
\lim_{n \rightarrow \infty} \frac{\ln \left | S(t)-S_n(t) \right |}{n} = 
-2 \,\text{{\em Re}} \left [ \Lambda_{\mu} (t) \right ]
\,.
$$
\label{errorTheorem}
\end{theorem}

\begin{proof}
Let $t \in {\mathbb C} \backslash {\mathbb R}_-$ be fixed.
We have
$$
S(t) = \sqrt{t} Z ({\mathbf s}) \quad \text{and} \quad S_n(t) = \sqrt{t} Z_n ({\mathbf s})\,.
$$
Hence, by lemma \ref{errorLemma},
\begin{equation}
\frac{\ln \left | S (t)-S_n(t)\right |}{n} = \frac{\ln |t|}{2 n} - \frac{\ln |u_n|}{n} + \frac{1}{n} \sum_{j=1}^n
\ln \left | Z \left ( {\mathscr T}^j {\mathbf s} \right ) \right |\,.
\label{errorEquation}
\end{equation} 
Consider the almost-sure limit of each of the three terms on the right-hand side of this equality as
$n \rightarrow \infty$.
The first term tends to zero.
By proposition \ref{furstenbergProposition}, the second term tends to
$$
-\text{Re} \left [ \Lambda_{\mu} (t) \right ]\,.
$$
Finally, it follows easily from the ergodic theorem that the third term tends to the same limit.
The result follows from a standard argument involving the
use of Fubini's theorem.

\end{proof}
\section{A numerical illustration}
\label{numericalSection}

Following the example of Dyson (1953), it is instructive to begin with an examination
of the (determistic) case where
$$
\mu = \mu^{\infty} := \delta_1\,.
$$
Set
$$
S^\infty (t) := \cfrac{1}{1 + \cfrac{t}{1 + \cfrac{t}{1+\cdots}}}\,, \quad t \in {\mathbb C} \backslash {\mathbb R}_-\,. 
$$
Then
$$
S^{\infty} (t) =
\frac{1}{1+\sqrt{1+4t}}
= \int_0^\infty \frac{\sigma^\infty (\rd x)}{1+xt}\,,
$$
where
$$
\sigma^{\infty} (\rd x) = \begin{cases}
\frac{1}{2 \pi} \sqrt{4/x-1} \,\rd x & \text{if $0 < x <4$} \\
0 & \text{if $x>4$}
\end{cases}\,.
$$
Note that $\mu^{\infty}$ has only {\em one} point of increase, and so the hypothesis of Proposition
\ref{essentialSpectrumProposition} is not satisfied. Indeed, for this choice of $\mu$, the spectrum
of the measure $\sigma$ is absolutely continuous.

In this case, the continued fraction (\ref{complexContinuedFraction}) reduces to
$$
Z^{\infty} = \frac{2}{\sqrt{1/t+4}+1/\sqrt{t}}
$$
and so the complex Lyapunov exponent is 
$$
\Lambda_{\mu^\infty} (t) = - \int_{\mathbb C} \ln z \,\nu_{\mu^\infty} (\rd z)
= - \int_{\mathbb C} \ln z \,\delta_{Z^\infty} (\rd z) = \ln \frac{\sqrt{1/t+4}+1/\sqrt{t}}{2}\,.
$$
In particular, an elementary calculation shows that
$$
N^{\infty} (\lambda) = 
\begin{cases}
1 - \frac{2}{\pi} \arccos{\frac{\sqrt{\lambda}}{2}} & \text{if $0 < \lambda < 4$}\\
1 & \text{if $\lambda > 4$}
\end{cases} \,.
$$

Next, let $a \in {\mathbb Z}_+$ and denote by $\mu^{a}$ the gamma distribution with
$b=1/a$. As Dyson remarked, this distribution has mean $1$ and variance $1/a$
and so we may, for large $a$, view it as a perturbation of
$\mu^{\infty}$. Indeed,
by using our explicit formula for $\Lambda_{\mu^a}$
together with the large-order expansions in Abramowitz \& Stegun (1964), \S 9.7, we find
$$
\Lambda_{\mu^a} (t) \sim \Lambda_{\mu^\infty} (t) + \frac{1}{a} \frac{1+8 t}{2(1+4t)} + O \left ( \frac{1}{a^2} \right ) \quad \text{as $a \rightarrow \infty$}\,, \;\; a \in {\mathbb Z}_+\,.
$$
Likewise, setting
$$
\beta := \arccos{\frac{\sqrt{\lambda}}{2}}\,, \quad 0 < \lambda < 4\,,
$$
and using the large-order expansions in Abramowitz \& Stegun (1964), \S 9.3, we obtain, for
$a \in {\mathbb Z}_+$,
$$
\varrho (\lambda) \sim \varrho^\infty (\lambda)
- \frac{1}{a^2} \frac{\cos \beta}{32\pi \sin^3 \beta} \left ( 13 + 38 \cot^2 \beta + 25 \cot^4 \beta 
\right ) + O \left ( \frac{1}{a^4} \right ) \quad \text{as $a \rightarrow \infty$}\,.
$$
This expansion breaks down at $\lambda=4$; as $a$ increases,
$\varrho(\lambda)$ diverges to infinity there but tends to zero exponentially fast for $\lambda>4$.

The following computations were performed in multiple-precison floating-point
arithmetic with the {\tt MAPLE} software package;
the eigenvalues and eigenvectors of the matrix $\mathscr{J}_n$ were calculated
by using the {\tt Eigenvals} function, which implements the QR algorithm.
Figure \ref{integratedDOSFigure} corresponds to the case
where $\mu=\mu^a$ with
$a=8$; it illustrates the convergence of the counting measure $N_n$ to the
integrated density of states $N$ as $n$ increases--- thus confirming the validity
of our formula for $N$.

\begin{figure}[htbp]
\vspace{7.5cm} 
\includegraphics{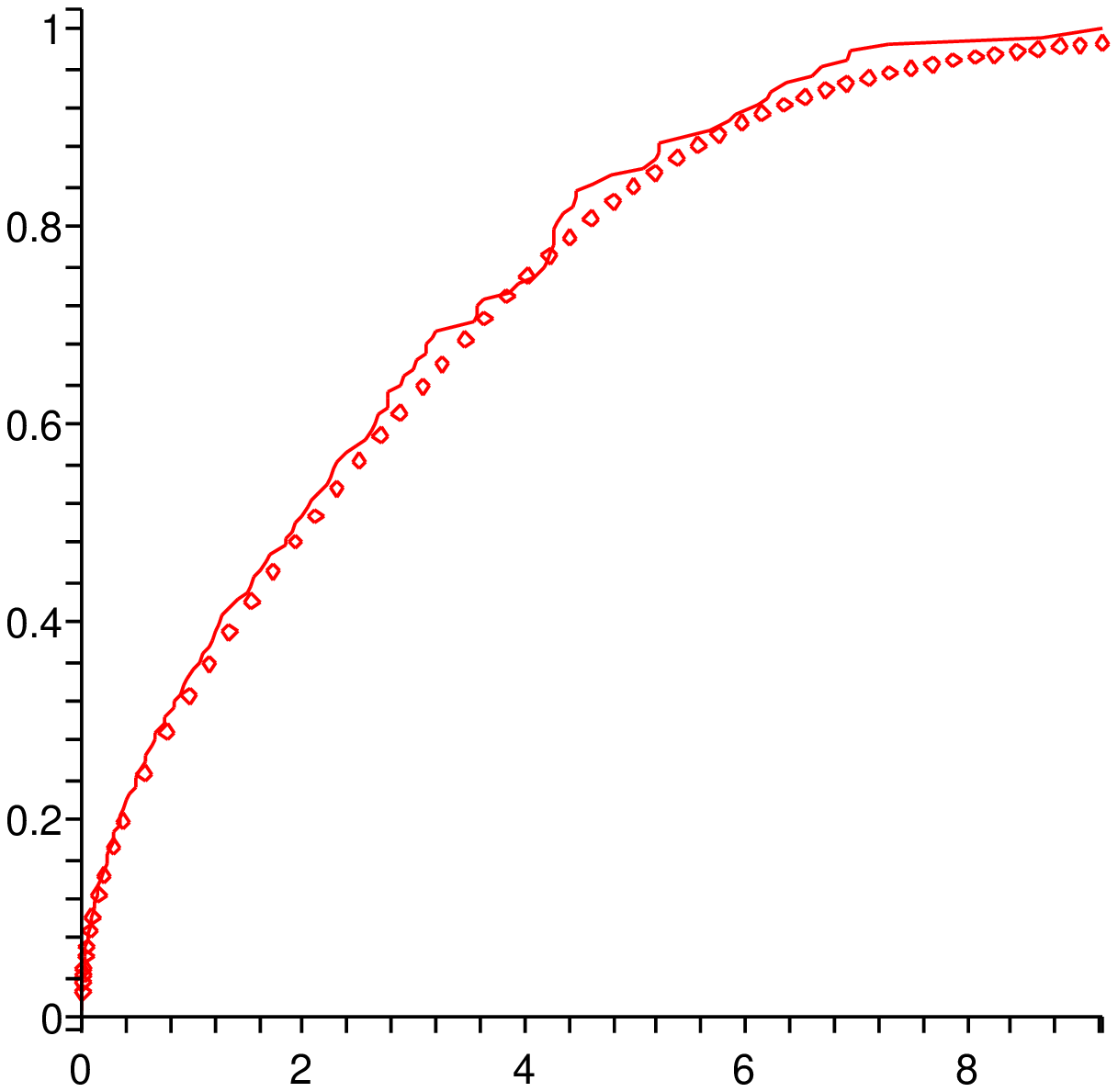}  
\includegraphics{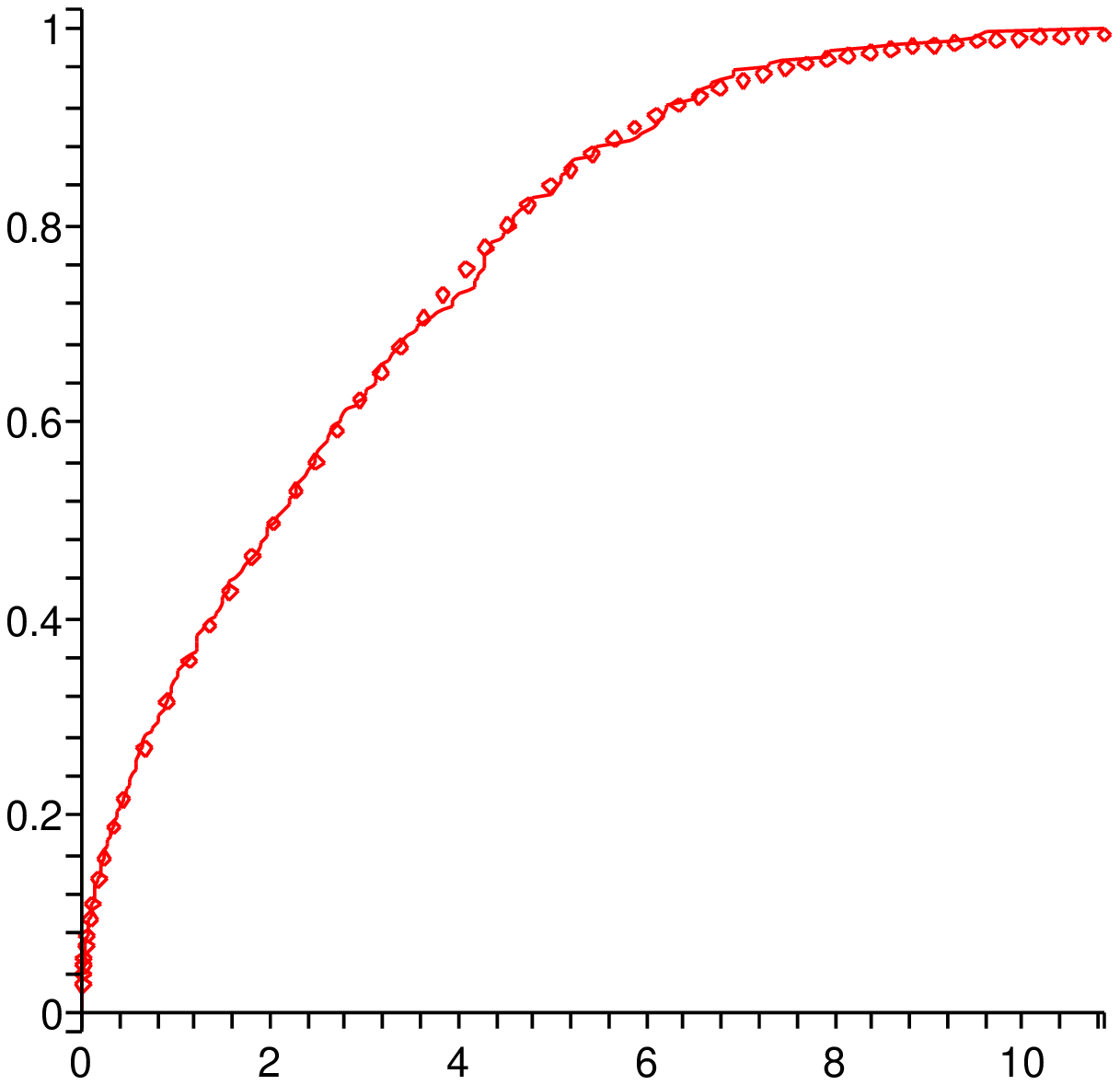} 
\begin{picture}(0,0) 
\put(-100,0){(i)}
\put(-60,10){$\lambda$} 
\put(100,0){(ii)} 
\put(140,10){$\lambda$}
\end{picture} 
\caption{The counting measure $N_n$ (solid line) for a particular realisation of the
sequence ${\mathbf s}$ when $\mu=\mu^a$ with $a=8$:
(i) $n=128$ and  (ii) $n=256$.
For comparison, points corresponding to values of the integrated density of states $N$
are also shown.}
\label{integratedDOSFigure} 
\end{figure}

While the Lyapunov exponent and the density of states
are non-random, the measure $\sigma$ {\em is} random.
We note that $\sigma^{\infty}$ is absolutely continuous whereas, for every
$a \in {\mathbb Z}_+$, almost every realisation of
$\sigma$ is singular. 
Figure \ref{singularMeasureFigure}
shows the approximation
$$
\int_0^{\lambda} \sigma_n (\rd \lambda')\,,
$$
of the integrated measure
for particular realisations corresponding to $a=8$ and $a=64$. 
Here $\sigma_n$ is the discrete measure defined by the quadrature formula
(\ref{quadratureFormula}), and $n=256$.

\begin{figure}[htbp]
\vspace{7.5cm} 
\includegraphics{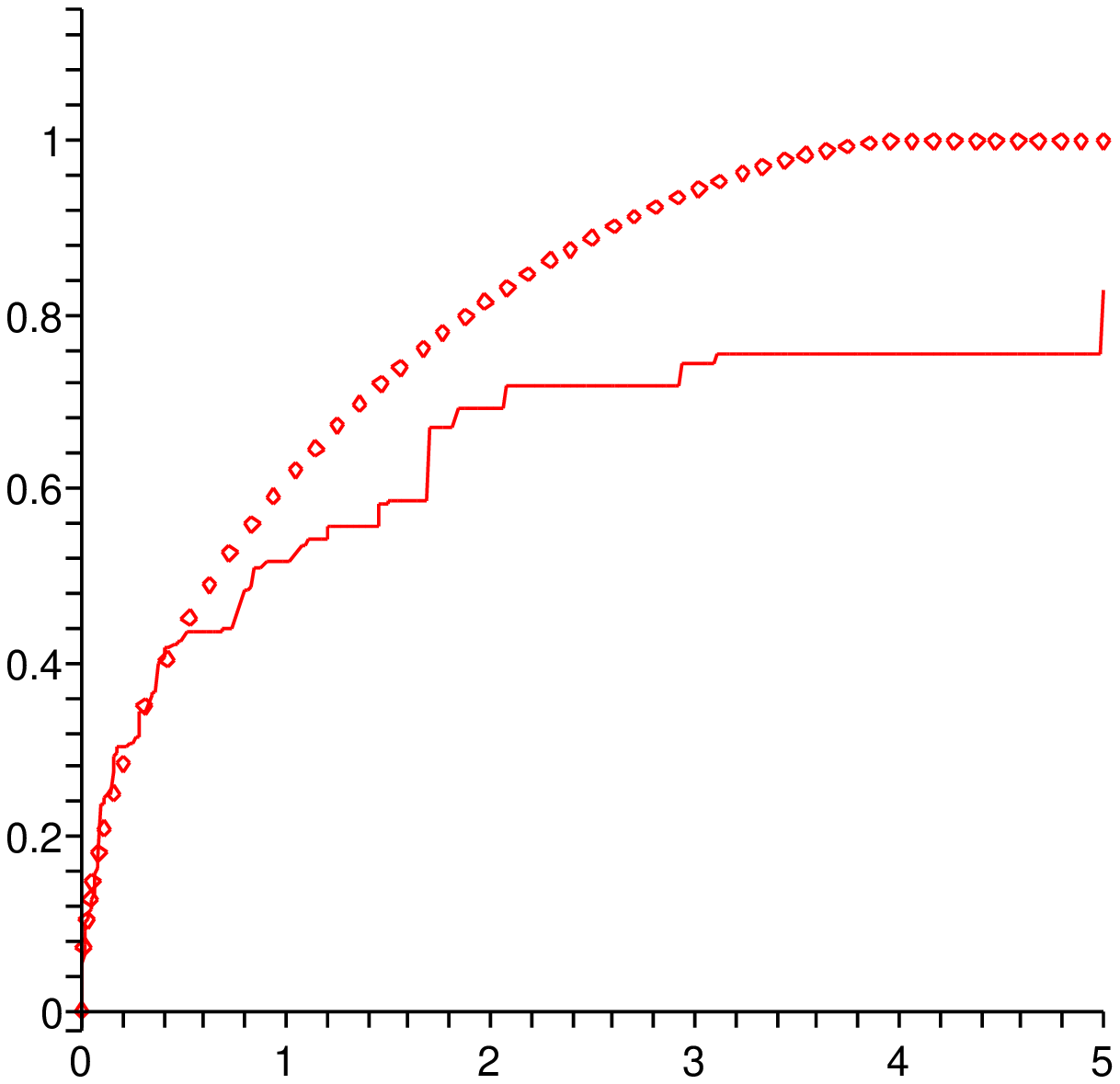}  
\includegraphics{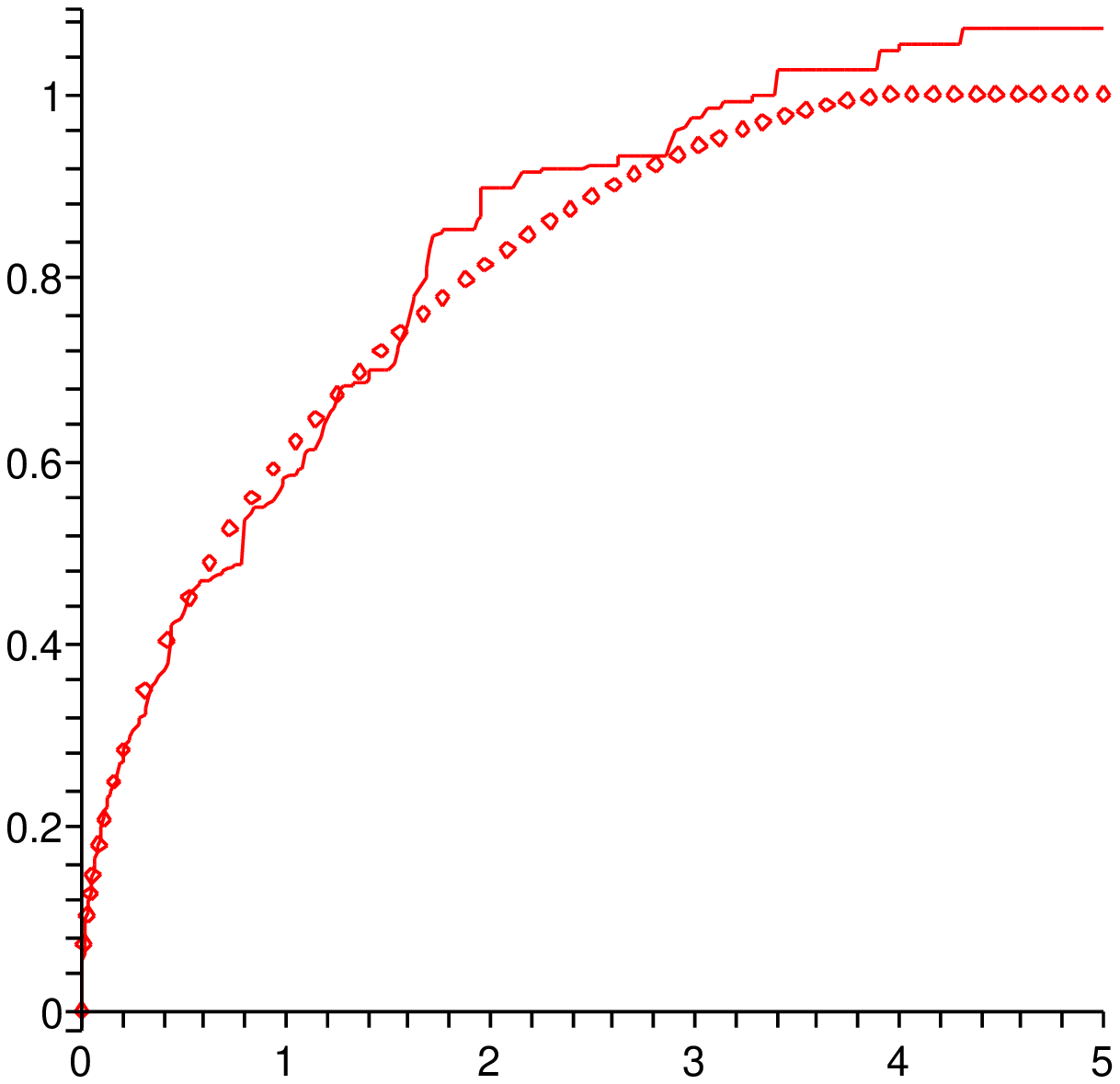}
\begin{picture}(0,0) 
\put(-100,0){(i)}
\put(-60,10){$\lambda$} 
\put(100,0){(ii)} 
\put(140,10){$\lambda$}
\end{picture}  
\caption{Plot (solid line) of 
$\int_0^\lambda \sigma_n (\rd \lambda')$, with $n=256$, corresponding to
particular realisations of ${\mathbf s}$
when $\mu=\mu^a$ where (i) $a=8$ and  (ii) $a=64$.
For comparison, points corresponding to values of the function $\int_0^\lambda \sigma^{\infty} (\rd \lambda')$
are also shown.}
\label{singularMeasureFigure} 
\end{figure}

\appendix

\section{Proof of proposition \ref{furstenbergProposition}}
\label{furstenbergAppendix}

As mentioned in the introduction, the theory of products of random matrices is a convenient
means of deducing the uniqueness of the invariant measure, as well as the positivity of the real part
of the complex Lyapunov exponent. For this purpose,
we shall have to deal with products of $2 \times 2$ matrices with real or complex entries;
so, in the following, ${\mathbb K}$ will stand for either ${\mathbb R}$
or ${\mathbb C}$. Set 
$$
\overline{\mathbb K} := {\mathbb K} \cup \{ \infty \}\,.
$$
We define an equivalence relation in the set of nonzero vectors in ${\mathbb K}^2$ via
$$
\begin{pmatrix}
u \\
v 
\end{pmatrix}
\sim 
\begin{pmatrix}
u' \\
v' 
\end{pmatrix}
\quad \text{if $\exists \; w \in {\mathbb K} \backslash \{0\}$ such that}\;
\begin{pmatrix}
u \\
v 
\end{pmatrix}
= w
\begin{pmatrix}
u' \\
v' 
\end{pmatrix}\,.
$$
The set of the equivalence classes is called the {\em projective space} ${P}({\mathbb K}^2)$.
Let
$$
\left [ \begin{pmatrix}
u \\
v
\end{pmatrix}
\right ] \in P \left ( {\mathbb K}^2 \right )\,.
$$
We shall identify this equivalence class with 
$$
z = {\mathcal P} \left ( 
\left [ \begin{pmatrix}
u \\
v
\end{pmatrix} \right ]
\right )
:= \begin{cases}
u/v & \text{if $v \ne 0$} \\
\infty & \text{otherwise}
\end{cases} \in \overline{\mathbb K}\,.
$$

The results of Furstenberg \& Kesten (1960) and  Furstenberg (1963)
concern the
typical asymptotic behaviour of the product of independent, identically
distributed random elements of some group acting on a compact topological space. In our particular
context, the relevant group is the subgroup of $\text{GL} \left ( 2,{\mathbb K}^2 \right )$ consisting
of $2 \times 2$ matrices with determinant $\pm 1$, and the topological space is $P \left ( {\mathbb K}^2 \right )$.
The invertible matrices
\begin{equation}
\mathcal A = \begin{pmatrix}
a & b \\
c & d
\end{pmatrix}\,,
\label{randomMatrix}
\end{equation}
are drawn at random 
from a distribution which we shall denote by $\tilde{\mu}$. 
The action of
the matrix $\mathcal A$ on the projective space can be expressed as
\begin{equation}
{\mathcal A} \cdot \begin{pmatrix} z \\ 1 \end{pmatrix}
= \begin{pmatrix} {\mathcal F} (z) \\ 1 \end{pmatrix} \,,
\label{groupAction}
\end{equation}
where ${\mathcal F} :\; \overline{\mathbb K} \rightarrow \overline{\mathbb K}$ is the
linear fractional transformation defined by
\begin{equation}
{\mathcal F} (z) = \mathcal P \left ( \mathcal A \begin{pmatrix} z \\ 1 \end{pmatrix} \right )
= \begin{cases}
\frac{a z + b}{c z + d} & \text{if $z \in {\mathbb K}$ and $c z + d \ne 0$} \\
\infty & \text{if $c \ne 0$ and $z=-d/c$} \\
a/c & \text{if $z = \infty$ and $c \ne 0$} \\
\infty & \text{if $c=0$ and $z = \infty$}
\end{cases}
\,.
\label{linearFractionalTransformation}
\end{equation}
Thus, we have an obvious connection between products 
$$
{\mathcal U}_n := {\mathcal A}_n \cdots  {\mathcal A}_2 {\mathcal A}_{1}
$$
of the ${\mathcal A}_n$ and the
Markov chain $\hat{\mathbf Z}$ such that
\begin{equation}
\hat{Z}_0 = z, \quad \hat{Z}_{n+1} = {\mathcal F}_n \left ( \hat{Z}_n \right )\,.
\label{generalMarkovChain}
\end{equation}
Indeed,
\begin{equation}
\hat{Z}_n =  {\mathcal U}_n  \cdot \begin{pmatrix}
z \\ 1
\end{pmatrix}\,.
\label{markovChainAsProduct}
\end{equation}
We are interested in the rate of growth of this product
as $n \rightarrow \infty$; this is quantified by  
the number
\begin{equation}
\gamma_{\tilde{\mu}} := \lim_{n \rightarrow \infty} \frac{1}{n} {\mathbb E} \left ( \ln |
{\mathcal U}_n  | \right ),
\label{lyapunovExponentDefinition}
\end{equation}
where $| \cdot |$ denotes some matrix norm.
This limit exists whenever 
\begin{equation}
{\mathbb E}(\log^+|{\mathcal A}|)<\infty\,.
\label{lyapunovCondition}
\end{equation}
In the literature on products of random matrices, the name ``Lyapunov exponent''
is usually reserved for $\gamma_{\tilde{\mu}}$; as we shall see, it is in fact
the real part of the complex Lyapunov exponent introduced earlier.

The following specialisation of Furstenberg's theorem will be the most useful for our purpose:
\begin{theorem}
Let the ${\mathcal A}_n$ be independent $2 \times 2$ random matrices with determinant $\pm 1$
and a common distribution $\tilde{\mu}$.
Suppose that there is no measure on $P \left ( {\mathbb K}^2 \right )$ that is invariant under the action of
the smallest subgroup of $\text{\em GL}(2,{\mathbb K}^2)$ generated by the support of $\tilde{\mu}$. Then
\begin{enumerate}
\item there 
exists a unique probability measure $\tilde{\nu}_{\tilde{\mu}}$ on $P \left ( {\mathbb K}^2 \right )$ that is invariant under the action
of matrices in the support of $\tilde{\mu}$. 
\item Let $z \in {\mathbb K}$. Then, for almost every realisation 
of the sequence $\{ \mathcal A_n \}_{n \in \mathbb N}$,
$$
\lim_{n \rightarrow \infty} \frac{1}{n} \ln  \left | {\mathcal U}_n \cdot \begin{pmatrix} z \\ 1 \end{pmatrix} \right |  
= \gamma_{\tilde{\mu}}\,.
$$
\item If condition (\ref{lyapunovCondition}) holds, then the number $\gamma_{\tilde{\mu}}$
is strictly positive and is given by the formula
\begin{equation}
\notag
\gamma_{\tilde{\mu}}
= \int_{P \left ( {\mathbb K}^2 \right )}
\int_{\text{\em GL} (2, \mathbb K)} \ln 
\frac{\left | \mathcal A \begin{pmatrix} z \\ 1 \end{pmatrix} \right |}{\left | \begin{pmatrix} z \\ 1 \end{pmatrix} \right |} \,
\tilde{\mu} \left ( \text{\em d} \mathcal A \right )
\,\tilde{\nu}_{\tilde{\mu}} \left ( \text{\em d} \begin{pmatrix} z \\ 1 \end{pmatrix} \right )\,.
\end{equation}
\end{enumerate}
\label{furstenbergTheorem}
\end{theorem} 

\begin{proof}
See Bougerol \& Lacroix (1985), Part A, theorem 4.4 on p. 32. The last formula appears in Part A, theorem 3.6 on p. 27
of the same reference.
\end{proof}

\begin{corollary}
Let $t \in {\mathbb C} \backslash {\mathbb R}_-$. Let the $s_n$ be independent 
random variables in ${\mathbb R}_+$ with a common distribution $\mu$
whose support contains at least two points. Set
$$
{\mathcal A}_n = \begin{pmatrix}
0 & 1 \\
1 & \frac{s_{n}}{\sqrt{t}} 
\end{pmatrix}\,.
$$
Then $\mu$ induces a probability distribution $\tilde{\mu}$
on $\text{\em GL}(2,{\mathbb C})$ such that
the hypothesis of theorem \ref{furstenbergTheorem} is fulfilled, and the image $\nu_{\mu}$
of the invariant measure $\tilde{\nu}_{\tilde{\mu}}$ under the map $\mathcal P$ is the unique stationary
distribution for the Markov chain $\hat{\mathbf Z}$. Furthermore, if we also assume that
$$
\int_{{\mathbb R}_+} s^\varepsilon\,\mu (\text{\em d} s) < \infty
$$
for some $\varepsilon > 0$, then we have the formula
$$
\gamma_{\tilde{\mu}}
= - \int_{\mathbb C} \ln |z|\,\nu_{\mu} (\text{\em d} z)\,.
$$
\label{furstenbergCorollary}
\end{corollary}

\begin{proof}
Same as that of Carmona \& Lacroix (1990), corollary IV.4.26. 
\end{proof}

\begin{corollary}
Let $u_n$ be the sequence defined by the recurrence relation
(\ref{differenceEquation}) with $u_0 \ne 0$. Under the same hypothesis on $\mu$,
for every $t \in {\mathbb C} \backslash {\mathbb R}_-$, for almost every realisation of the sequence
${\mathbf s}$, we have
$$
\lim_{n \rightarrow \infty} \frac{\ln u_n}{n} = \Lambda_{\mu} (t)
$$
independently of the starting value $u_0$. 
\label{markovChainCorollary}
\end{corollary}

\begin{proof}
This follows from the equality
$$
\frac{\ln u_n}{n}  = \frac{\ln u_0}{n} - \frac{1}{n} \sum_{j=1}^n \ln \hat{Z}_j\,,
$$
the definition of $\Lambda_{\mu}$,
and the uniqueness of the stationary measure for the Markov chain $\hat{\mathbf Z}$.
\end{proof}

For reasons that will become clear when we discuss the nature of the spectrum of the
measure $\sigma$, we also need to consider the case where 
$$
t = -x + \ri 0 \pm\,, \quad x > 0\,.
$$
In this case,
$$
{\mathcal A}_n = \begin{pmatrix}
0 & 1 \\
1 & \mp \ri \frac{s_n}{\sqrt{x}} 
\end{pmatrix}
$$
and so the recurrence defining the Markov chain $\hat{\mathbf Z}$ is
$$
\hat{Z}_{n+1} = \frac{1}{\hat{Z}_n + \mp \ri \frac{s_{n+1}}{\sqrt{x}}}\,.
$$
Once again, theorem \ref{furstenbergTheorem} implies the existence of a unique
stationary measure $\mu_{\nu}$. We remark that
this stationary measure is concentrated on the imaginary axis. Indeed,
define $\hat{Y}_n$ via
$$
\hat{Z}_{n} = \ri \hat{Y}_n\,.
$$
Then 
$$
\hat{Y}_{n+1} = \frac{-1}{\hat{Y}_n + \mp \frac{s_{n+1}}{\sqrt{x}}}\,.
$$
The resulting Markov chain $\hat{\mathbf Y}= (\hat{Y}_0\,,\hat{Y}_1\,, \ldots )$ 
corresponds to the product of the real (Schr\"{o}dinger)
matrices
$$
\begin{pmatrix}
0 & -1 \\
1 & \frac{s_{n}}{\sqrt{x}} 
\end{pmatrix}\,.
$$
By applying theorem \ref{furstenbergTheorem} with ${\mathbb K} = {\mathbb R}$, we deduce
that $\hat{\mathbf Y}$ has a unique stationary measure supported on ${\mathbb R}$. The fact
that the measure $\mu_{\nu}$ must be concentrated
on the imaginary axis then follows by uniqueness.

\begin{corollary}
Let $u_n$ be the sequence defined by the recurrence relation
(\ref{differenceEquation}) with $u_0 \ne 0$ and $t = -x + \ri 0 \pm$. Under the same hypothesis on $\mu$,
for every $x \in {\mathbb R}_+$, for almost every realisation of the sequence
${\mathbf s}$, we have
$$
\lim_{n \rightarrow \infty} \frac{\ln u_n}{n} = \Lambda_{\mu} ( -x + \ri 0 \pm  )
$$
independently of the starting value $u_0$.
\label{markovChainCorollaryOnTheNegativeAxis}
\end{corollary}

Finally, the positivity of $\gamma_{\tilde{\mu}}$, asserted in theorem \ref{furstenbergTheorem},
leads to the
\begin{corollary}
Under the same hypothesis on $\mu$, for every $t \in {\mathbb C} \backslash {\mathbb R}_-$,
$$
\Real \left [ \Lambda_{\mu} (t) \right ] > 0\,.
$$
Also, for every $x \in {\mathbb R}_+$,
$$
\Real \left [ \Lambda_{\mu} (-x + \ri 0 \pm) \right ] > 0\,. 
$$
\label{positivityCorollary}
\end{corollary}

\begin{proof}
Take the real part in equation (\ref{characteristicExponent}) and use corollary \ref{furstenbergCorollary}\,.
\end{proof}

The statements made in corollaries \ref{furstenbergCorollary}-\ref{positivityCorollary} are of the
form: 
\begin{quote}
{\em Let $t$ be fixed, then for almost every realisation of the sequence ${\mathbf s}$, etc.}
\end{quote}
But the proof of proposition \ref{furstenbergProposition} follows easily from them by
a well-known argument based on the use of Fubini's theorem; see, for instance, Ishii (1973).

\section{The complex Lyapunov exponent for the gamma distribution}
\label{characteristicExponentAppendix}
In this appendix, we derive the formula (\ref{characteristicExponentFormula})
for the complex Lyapunov exponent when $\mu$ is the gamma distribution.
For this purpose, it is convenient to 
adopt the notation used in Marklof {\em et al.} (2005); so we set
$$
a_n := \frac{s_n}{\sqrt{|t|}} \quad \text{and} \quad \alpha := -\frac{\arg{t}}{2} \in [-\pi/2,\,\pi/2]\,. 
$$
Then 
$$
Z = \cfrac{1}{a_1 \re^{\ri \alpha}+\cfrac{1}{a_2 \re^{\ri \alpha}+\cfrac{1}{a_3 \re^{\ri \alpha}+ \cdots}}}\,.
$$
The random variable $Z$ takes values in the set
$$
S_\alpha := \left \{ z \in \mathbb C\,:\; | \arg z | \le |\alpha| \right \}\,.
$$
In Marklof {\em et al.} (2005), we showed that, if the $a_n$ are gamma-distributed with parameters $p$ and $s$, then
the probability density function of $Z$ is given explicitly by
\begin{multline}
\label{densityFunction}
f_\alpha (z) = 
\frac{\sin(2 |\alpha|)}{\left | 2 K_p \left ( 2 \re^{\ri \alpha}/s \right ) \right |^{2}} 
\frac{1 }{r^2 \sin^2(\alpha+\theta)}
\left [ \frac{\sin(\alpha-\theta)}{\sin(\alpha+\theta)}
\right ]^{p-1} \\
\times \exp \left \{ -\frac{\sin(2 \alpha)}{s}
\left [ \frac{1}{r
\sin(\alpha-\theta)}+ \frac{r}{\sin(\alpha+\theta)} \right ] \right \}\,,
\end{multline}
where $r = |z|$ and $\theta = \arg z$. This is valid for $|\alpha | < \pi/2$--- which is all
we need for the purpose of calculating the complex Lyapunov exponent, since the cases 
$\alpha =\pm \pi/2$ may be obtained by letting $\alpha$ tend to the appropriate limit.
In Marklof {\em et al.} (2005), we derived the formula
$$
-\int_{S_\alpha} \ln |z| f_\alpha (z) \,\rd z = \Real \left [ \frac{\partial_p K_p \left (\frac{2}{s}
\re^{-\ri \alpha} \right )}{K_p \left (\frac{2}{s}
\re^{-\ri \alpha} \right )} \right ]\,.
$$
So there only remains to show that
\begin{equation}
-\int_{S_\alpha} \ln \left ( \re^{\ri \arg z}\right )\, f_\alpha (z) \,\rd z = \ri \Imag \left  [ \frac{\partial_p K_p \left (\frac{2}{s}
\re^{-\ri \alpha} \right )}{K_p \left (\frac{2}{s}
\re^{-\ri \alpha} \right )} \right ]\,.
\label{imaginaryPartFormula}
\end{equation}

We shall need the 
\begin{lemma}
Let $w$ and $W$ be two complex numbers with positive real part. Then
\begin{multline}
\notag
K_p(w) \,\partial_p K_p (W) 
= -\frac{1}{2} \int_{-\infty}^\infty \re^{-2 p u} 
\int_0^\infty \left ( u+\ln \frac{2v}{w \re^u + W \re^{-u}} \right )\\
\times \exp \left \{ -v - \frac{w^2+2 wW \cosh(2 u)+W}{4 v} \right \} \,\frac{\rd v}{v} \,\rd u\,.
\end{multline}
\label{besselLemma}
\end{lemma}

\begin{proof}
The Bessel function of the second kind has the integral representation (see Watson 1966, \S 6.22)
$$
K_p(z) = \frac{1}{2} \int_{-\infty}^\infty \re^{-z \cosh x - p x} \,\rd x\,.
$$
Hence
$$
K_p (W) \partial_p K_p (w) = -\frac{1}{4} \int_{-\infty}^\infty \int_{-\infty}^\infty
\exp \left \{-W \cosh y - p y - w \cosh x -p x \right \} \,x \,\rd x \rd y\,.
$$
We make the change of variables
$$
x = X+Y \quad \text{and} \quad y = X-Y\,.
$$
Then, after some re-arrangement,
\begin{multline}
\notag
K_p (W) \,\partial_p K_p (w) = - \frac{1}{2} \int_{-\infty}^\infty \re^{- 2 p X} \\
\times \int_{-\infty}^\infty \exp \left \{ - \re^Y \frac{W \re^{-X} + w \re^X}{2} - \re^{-Y}
\frac{W \re^X + w \re^{-X}}{2} \right \} (X+Y) \rd Y \rd X\,.
\end{multline}
Make the substitution
$$
v = \re^Y \frac{W \re^{-X}+w \re^X}{2}
$$
in the second integral. Then
\begin{multline}
\notag
K_p (W) \,\partial_p K_p (w) = - \frac{1}{2} \int_{-\infty}^\infty \re^{- 2 p X} \\
\times \int_{-\infty}^\infty \exp \left \{ -v - 
\frac{W^2+ 2 wW \cosh(2X) + w^2}{4 v} \right \} \left ( X+ \ln \frac{2v}{W \re^{-X}+w \re^X}
\right ) \frac{\rd v}{v} \rd X\,.
\end{multline}
The desired result follows after we set $X=u$ and take the imaginary part.
\end{proof}

Returning to the proof of equation (\ref{imaginaryPartFormula}), let us write
\begin{equation}
\Imag \int_{S_\alpha} \ln \left ( \re^{-\ri \arg z} \right ) \,f_\alpha (z) \,\rd z=
\frac{1}{4 \left | K_p \left ( \frac{2}{s} \re^{-\ri \alpha} \right ) \right |^2} \,{\mathcal I}\,,
\label{definitionOfI}
\end{equation}
where
\begin{multline}
\notag
{\mathcal I} :=  \Imag
\int_{-\alpha}^{\alpha} \int_0^\infty 
\ln \left ( \re^{-\ri \theta} \right )
\frac{\sin ( 2 |\alpha |)}{r^2 \sin^2(\alpha+\theta)}
\left [ \frac{\sin(\alpha-\theta)}{\sin(\alpha+\theta)}
\right ]^{p-1} \\
\times \exp \left \{ -\frac{\sin(2 \alpha)}{s}
\left [ \frac{1}{r
\sin(\alpha-\theta)}+ \frac{r}{\sin(\alpha+\theta)} \right ] \right \}\, \rd r \rd \theta\,.
\end{multline}

Replace the variable $r$ by
$$
v = \frac{r \sin (2 \alpha)}{s\,\sin (\alpha+\theta)}
$$ 
and then the variable $\theta$ by 
$$
t = \frac{\sin(\alpha-\theta)}{\sin (\alpha+\theta)}\,.
$$
A straightforward calculation leads to
\begin{equation}
\label{integral}
{\mathcal I} 
=  \text{Im} \int_0^\infty 
\ln \left ( \frac{\sqrt{t \varphi(t)}}{\re^{\ri \alpha}+t \re^{-\ri \alpha}} \right ) \,t^{p-1}
\int_0^\infty \exp \left \{ -v - \frac{4}{s^2} \frac{\varphi(t)}{4 v} \right \} \frac{\rd v}{v} \rd t\,,
\end{equation}
where
$$
\varphi(t) = t + \frac{1}{t} + 2 \cos \alpha\,.
$$
Set 
$$
t = \re^{-2 u}\,. 
$$
Then equation (\ref{integral}) becomes
\begin{multline}
\label{integral2}
{\mathcal I}
=  -2\,\text{Im} \int_{-\infty}^\infty 
\ln \left ( \re^{u+\ri \alpha}+ \re^{-u-\ri \alpha} \right ) \,\re^{-2 p u} \\
\times \int_0^\infty \exp \left \{ -v - \frac{4}{s^2} \frac{\re^{-\ri 2 \alpha}+2 \cosh(2 u) +\re^{\ri 2 \alpha} }{4 v} \right \} \frac{\rd v}{v} \rd u\,.
\end{multline}
We use lemma \ref{besselLemma} with $W = \overline{w} = 2 \re^{-\ri \alpha}/s$ to obtain
\begin{multline}
{\mathcal I} = -4 \,\Imag \left [ K_p \left ( \frac{2}{s} \re^{-\ri \alpha} \right ) 
\,\partial_p K_p \left ( \frac{2}{s} \re^{\ri \alpha} \right ) \right ] \\
= 4 \,\Imag \left [ K_p \left ( \frac{2}{s} \re^{\ri \alpha} \right ) 
\,\partial_p K_p \left ( \frac{2}{s} \re^{-\ri \alpha} \right ) \right ]\,.
\end{multline}
The result is then an immediate consequence of equation (\ref{definitionOfI}).

\section{A result of Goldsheid \& Khoruzhenko (2005)}
\label{goldsheidAppendix}
In order to deduce the existence of the integrated density of states from proposition
\ref{goldsheidKhoruzhenkoProposition}, we shall need some technical results contained
in Goldsheid \& Khoruzhenko (2005).

Let $\{ \mathscr{A}_n \}_{n \in {\mathbb Z}_+}$ be a deterministic sequence of square matrices of increasing dimension $n$, and
let
$$
p_n (z) = \frac{1}{n} \ln \left | \text{det} \left ( \mathscr{A}_n - z \mathscr{I}_n \right ) \right |
= \int_{\mathbb C} \ln \left |w-z \right | \kappa_n (\rd w)\,,
$$
where $\mathscr{I}_n$ is the $n \times n$ identity matrix and 
$$
\kappa_n = \frac{1}{2\pi} \Delta p_n
$$
is the normalised eigenvalue counting measure of $\mathscr{A}_n$. Define
\begin{equation}
\tau_R := \overline{\lim_{n \rightarrow \infty}} \int_{|w| \ge R} \ln |w | \,\kappa_n (\rd w)\,, \quad R \ge 1\,.
\label{tauDefinition}
\end{equation}

\begin{proposition}
Assume that there is a function $p\,:\;{\mathbb C} \rightarrow [-\infty,\,\infty)$ such that
$$
p_n(z) \xrightarrow[n \rightarrow \infty]{} p(z) \quad \text{for Lebesgue-almost every $z \in{\mathbb C}$}\,.
$$
If $\tau_1 < \infty$, then
\begin{enumerate}
\item $p$ is locally integrable.
\item The measure
$$
\kappa := \frac{1}{2 \pi} \Delta p
$$
is a probability measure.
\item We have
$$
\int_{|w| \ge 1} \ln |w|\, \kappa (\text{\em d} w) \le \tau_1 < \infty
$$
and the sequence $\{ \kappa_n \}_{n \in {\mathbb Z}_+}$ converges weakly to $\kappa$.
\end{enumerate}
\label{GKProposition13}
\end{proposition}

\begin{proof}
See Goldsheid \& Khoruzhenko (2005), proposition 1.3.
\end{proof}

To apply this result in our context, we set $\mathscr{A}_n = \mathscr{J}_n$.
Proposition \ref{goldsheidKhoruzhenkoProposition} then takes care of the first assumption. There remains
to show the finiteness of $\tau_1$. Goldsheid \& Khoruzhenko
remark that the following inequalities hold:
$$
\tau_1 \le \overline{\lim_{n \rightarrow \infty}} \frac{1}{2n} \text{tr} \ln \left ( \mathscr{I}_n
+ \mathscr{J}_n \mathscr{J}_n^\ast \right )
$$
and
\begin{equation}
\frac{1}{n} \text{tr} \ln \left ( \mathscr{I}_n
+ \mathscr{J}_n \mathscr{J}_n^\ast \right ) 
\le \frac{\alpha}{n} \sum_{j=0}^{n-1} \ln \left ( 1+\beta |{\mathbf r}_j |^2 \right )  \,,
\label{goldsheidKhoruzhenkoInequality}
\end{equation}
for some positive constants $\alpha$ and $\beta$ independent of $n$ and ${\mathbf r}_j$, the
$j$th row of $\mathscr{J}_n$. Now,
\begin{multline}
\notag
| {\mathbf r}_n |^2 = h_{n-1}^2 + v_n^2 + h_n^2 \\ = \frac{1}{s_{2n+2}^2 s_{2n+1} s_{2n+3}}
+ \frac{1}{s_{2n+1}^2} \left ( \frac{1}{s_{2n+2}}+\frac{1}{s_{2n+3}} \right )^2 
+ \frac{1}{s_{2n}^2 s_{2n-1} s_{2n+1}} 
\end{multline}
and, so by repeated use of the elementary inequality
$$
a b \le \frac{1}{2} \left ( a^2 + b^2 \right )\,,
$$
we obtain readily
$$
| {\mathbf r}_n |^2 \le \frac{5}{2} \left ( \frac{1}{s_{2n-1}^4} + \cdots + \frac{1}{s_{2n+3}^4} \right )\,. 
$$
It follows that
\begin{multline}
\notag
\ln \left ( 1+ \beta | {\mathbf r}_n |^2 \right ) \le 
\ln \left ( 1+ \frac{25}{2} \beta  \max \left \{ s_{2n-1}^{-4},\,\ldots,\, s_{2n+3}^{-4} \right \} \right ) \\
\le \ln \left ( 1+ \frac{25}{2} \beta s_{2n-1}^{-4}  \right ) + \cdots +
\ln \left ( 1+ \frac{25}{2} \beta s_{2n+3}^{-4}  \right )\,.
\end{multline}
We deduce from the ergodic theorem that the right-hand side of equation (\ref{goldsheidKhoruzhenkoInequality})
is bounded independently of $n$ provided that
$$
\int_{{\mathbb R}_+} \ln \left ( 1 + \frac{25}{2} \frac{\beta}{s^4} \right ) \,\mu(\rd s) < \infty\,.
$$
This last inequality follows easily from our hypothesis that
$$
\int_{{\mathbb R}_+} \left | \ln s \right |\,\mu(\rd s) < \infty \quad \text{and} \quad 
\int_{{\mathbb R}_+} s^\varepsilon \,\mu(\rd s) < \infty \,.
$$

\end{document}